\documentclass[11pt]{amsart}
\usepackage{amssymb}
\usepackage[english]{babel}
\usepackage{float}
\usepackage{rotating}
\usepackage{booktabs}   
\usepackage{latexsym}
\usepackage{amsmath, amsfonts}
\usepackage{bbm}
\usepackage{amsthm,amssymb,fancyhdr,verbatim,graphicx}
\usepackage[usenames,dvipsnames]{color}
\definecolor{darkred}{rgb}{0.4,0.1,0.1}
\definecolor{darkblue}{rgb}{0.1,0.1,0.4}
\usepackage[colorlinks=true,linkcolor =darkred, citecolor=darkblue]{hyperref}
\theoremstyle{plain}

\newtheorem{thm}{Theorem}[section]

\newtheorem*{thm*}{Theorem}

\newtheorem{prop}[thm]{Proposition}

\newtheorem*{cor*}{Corollary}

\theoremstyle{remark}
\newtheorem{remark}[thm]{Remark}
\theoremstyle{plain}

\newcommand{\be}{\begin{equation}}
\newcommand{\ee}{\end{equation}}
\newcommand{\beu}{\begin{equation*}}
\newcommand{\eeu}{\end{equation*}}
\newcommand{\besu}{\begin{equation*}
\begin{aligned}}
\newcommand{\eesu}{\end{aligned}
\end{equation*}}
\newcommand{\bes}{\begin{equation}
\begin{aligned}}

\newcommand{\ees}{\end{aligned}
\end{equation}}

\newcommand\fra{\mathfrak a}

\newcommand\sess{\sigma_{\rm ess}}

\newcommand\void[1]{}

\def\sess{\sigma_{\rm ess}}

\def\frb{{\mathfrak b}}

      \def\dC{{\mathbb C}}

      \def\dR{{\mathbb R}}

\newcommand{\dom}{\mathrm{dom}\,}

\numberwithin{equation}{section}

\title[$\delta$-interactions on angles and crossing lines]{Note on 2D Schr\"odinger operators with $\delta$-interactions on angles and crossing lines}

\author{Vladimir Lotoreichik}
\address{Technische Universit\"{a}t Graz,
Institut f\"{u}r Numerische Mathematik\\
Steyrergasse 30,
8010 Graz, Austria}
\email{lotoreichik@math.tugraz.at}

\begin{document}

\maketitle


\begin{abstract}
In this note we sharpen the lower bound from \cite{LLP10} on the spectrum of the 2D Schr\"odinger operator with a $\delta$-interaction supported on a planar angle.
Using the same method we obtain the lower bound on the spectrum of the 2D Schr\"odinger operator with a $\delta$-interaction supported on crossing straight lines. The latter operators
arise in the three-body quantum problem with $\delta$-interactions between particles.
\end{abstract}

\section{Introduction}

Self-adjoint Schr\"odinger operators with $\delta$-interactions supported on sufficiently regular hypersurfaces can be defined via closed, densely defined, symmetric and lower-semibounded quadratic forms using the first representation theorem, see \cite{BEKS94} and also \cite{BLL13}. 

\subsection*{$\delta$-interactions on angles}
In our first model the support of the $\delta$-interaction is the set $\Sigma_\varphi\subset\dR^2$, which consists of two rays meeting at the common origin  and constituting the angle $\varphi\in(0,\pi]$ as in Figure~\ref{fig}.
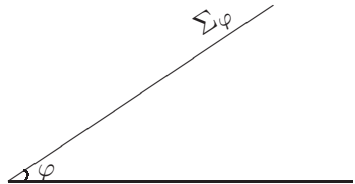
\begin{figure}[H]
\begin{center}
\begin{picture}(150,80)
\put(0,10){\line(1,0){130}}
\put(0,10){\line(3,2){100}}
\put(69,65){\begin{turn}{33.3}$\Sigma_\varphi$\end{turn}}
\qbezier(6,14)(8,13)(6,10)
\put(11,12){\small{$ \varphi$}}
\end{picture}
\end{center}
\caption{The angle $\Sigma_\varphi$ of degree $\varphi\in(0,\pi]$.}
\label{fig}
\end{figure}
\noindent The quadratic form in $L^2(\dR^2)$
\begin{equation}
\label{qform}
\fra_{\varphi}[f] := \|\nabla f\|_{L^2(\dR^2;\dC^2)}^2 - \alpha \|f|_{\Sigma_\varphi}\|_{L^2(\Sigma_\varphi)}^2,\qquad \dom\fra_{\varphi} := H^1(\dR^2),
\end{equation}
is closed, densely defined, symmetric and lower-semibounded, where $f|_{\Sigma_\varphi}$ is the trace of $f$ on $\Sigma_\varphi$, and the constant $\alpha > 0$ is
called the strength of interaction.
The corresponding self-adjoint operator in $L^2(\dR^2)$
we denote by $A_{\varphi}$. Known spectral properties of this operator
include explicit representation of the essential spectrum $\sess(A_\varphi) = [-\alpha^2/4,+\infty)$ and some information on the discrete spectrum: $\sharp\sigma_{\rm d}(A_\varphi) \ge 1$ if and only if $\varphi \ne \pi$.
These two statements can be deduced from more general results by Exner and Ichinose \cite{EI01}. They are complemented by Exner and Nem\v{c}ov\'{a} in \cite{EN03} with 
the limiting property $\sharp\sigma_{\rm d}(A_\varphi) \rightarrow +\infty$ as $\varphi \rightarrow 0+$. 

In \cite{LLP10} the author obtained jointly with Igor Lobanov and Igor Yu. Popov 
a general result, which implies the lower bound on the spectrum of $A_\varphi$
\begin{equation}
\label{LLPbound}
\inf\sigma(A_\varphi) \ge -\frac{\alpha^2}{4\sin^2(\varphi/2)}.
\end{equation}
This bound is close to optimal for $\varphi$ close to $\pi$, whereas in the limit
$\varphi\rightarrow 0+$ the bound tends to $-\infty$. In the present note we sharpen this bound. Namely, we obtain
\begin{equation}
\label{bound1}
\inf\sigma(A_\varphi) \ge -\frac{\alpha^2}{(1+\sin(\varphi/2))^2}.
\end{equation}
The new bound yields that the operators $A_\varphi$ are uniformly
lower-semibounded with respect to $\varphi$ and
\[
\inf\sigma(A_\varphi)\ge -\alpha^2
\]
holds for all $\varphi \in(0,\pi]$.
This observation agrees well with physical expectations.
Note that separation of variables yields that $\inf\sigma(A_\pi) = -\alpha^2/4$
and in this case the lower bound in \eqref{bound1} coincides  with the exact
spectral bottom.

For sufficiently sharp angles upper bounds on $\inf\sigma(A_\varphi)$ were obtained by Brown, Eastham and Wood in \cite{BEW08}.
See also Open Problem 7.3 in \cite{E08} related to the discrete spectrum
of $A_\varphi$ for $\varphi$ close to $\pi$.

\subsection*{$\delta$-interactions on crossing straight lines}
We also consider an analogous model with the $\delta$-interaction supported on the set $\Gamma_\varphi = \Gamma_1\cup \Gamma_2$,  where $\Gamma_1$
and $\Gamma_2$ are two straight lines, which cross at the angle $\varphi\in(0,\pi)$ as in Figure~\ref{fig2}.
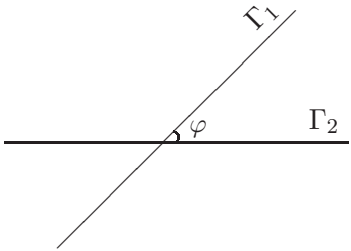
\begin{figure}[H]
\begin{center}
\begin{picture}(150,100)
\put(0,50){\line(1,0){130}}
\put(20,10){\line(1,1){90}}
\put(90,92){\begin{turn}{45.0}$\Gamma_1$\end{turn}}
\put(115,55){$\Gamma_2$}
\qbezier(64,54)(67,53)(65,50)
\put(70,54){\small{$ \varphi$}}
\end{picture}
\end{center}
\caption{The straight lines $\Gamma_1$ and $\Gamma_2$ crossing at the angle of degree $\varphi\in(0,\pi)$.}
\label{fig2}
\end{figure}
\noindent The corresponding self-adjoint operator $B_{\varphi}$ in $L^2(\dR^2)$ can be defined via the closed, densely defined, symmetric and lower-semibounded  quadratic form 
\begin{equation}
\label{qform2}
\frb_{\varphi}[f] := \|\nabla f\|_{L^2(\dR^2;\dC^2)}^2 - \alpha \|f|_{\Gamma_\varphi}\|_{L^2(\Gamma_\varphi)}^2,\qquad \dom\frb_\varphi := H^1(\dR^2),
\end{equation}
in $L^2(\dR^2)$, where $\alpha > 0$ is the strength of interaction.
According to \cite{EN03} it is known that $\sigma_{\rm ess}(B_\varphi) = [-\alpha^2/4,+\infty)$
and that $\sharp\sigma_{\rm d}(B_\varphi) \ge 1$.

In this note we obtain the lower bound 
\begin{equation}
\label{bound2}
\inf\sigma(B_\varphi) \ge -\frac{\alpha^2}{1+\sin\varphi},
\end{equation}
using the same method as for the operator $A_\varphi$.
Separation of variables yields 
$\inf\sigma(B_{\pi/2}) = -\alpha^2/2$, and in this case the lower bound in the  estimate \eqref{bound2}  coincides with the exact spectral bottom.

Upper bounds on $\inf\sigma(B_\varphi)$ were obtained in \cite{BEW08, BEW09}. 
The operators of the type $B_\varphi$ arise in the one-dimensional quantum three-body problem after excluding the center of mass, see Cornean, Duclos and Ricaud \cite{CDR06, CDR08} and the references therein. 

We want to stress that our proofs are of elementary nature
and we do not use any reduction to integral operators acting on interaction supports $\Sigma_\varphi$ and $\Gamma_\varphi$.

\section{Sobolev spaces on wedges}
In this section $\Omega\subset\dR^2$ is a wedge with the angle
of degree $\varphi\in (0,2\pi)$. 
The Sobolev space $H^1(\Omega)$
is defined as usual, see \cite[Chapter 3]{McLean}. 
For any $f \in H^1(\Omega)$ the trace $f|_{\partial\Omega}\in L^2(\partial\Omega)$ is well-defined as in \cite[Chapter 3]{McLean} and \cite{M87}.
\begin{prop}\cite[Lemma 2.6]{LP08}
\label{wedge1}
Let $\Omega$ be a wedge with angle of degree $\varphi \in(0,\pi]$.
Then for any $f\in H^1(\Omega)$ the estimate
\[
\|\nabla f\|^2_{L^2(\Omega;\dC^2)} - \gamma\|f|_{\partial\Omega}\|^2_{L^2(\partial\Omega)} \ge -  \frac{\gamma^2}{\sin^2(\varphi/2)}\|f\|_{L^2(\Omega)}^2
\]
holds for all $\gamma >0$.
\end{prop}
\begin{prop}\cite[Lemma 2.8]{LP08}
\label{wedge2}
Let $\Omega$ be a wedge with angle of degree $\varphi \in(\pi,2\pi)$.
Then for any $f\in H^1(\Omega)$ the estimate
\[
\|\nabla f\|^2_{L^2(\Omega;\dC^2)} - \gamma\|f|_{\partial\Omega}\|^2_{L^2(\partial\Omega)} \ge -  \gamma^2\|f\|_{L^2(\Omega)}^2
\]
holds for all $\gamma > 0$.
\end{prop}
Propositions~\ref{wedge1} and \ref{wedge2} are variational equivalents of spectral
results from \cite{LP08}.

\section{A lower bound on the spectrum of $A_\varphi$}
In the next theorem we sharpen the bound \eqref{LLPbound} using
only properties of the Sobolev space $H^1$ on wedges and some optimization.
\begin{thm}
\label{thm:lowerbound}
Let the self-adjoint operator $A_\varphi$ be associated with the quadratic
form given in \eqref{qform}. Then the estimate
\[
\inf\sigma(A_\varphi) \ge -\frac{\alpha^2}{\big(1+\sin(\varphi/2)\big)^2}
\]
holds.
\end{thm}
\begin{proof}
The angle $\Sigma_\varphi$ separates the Euclidean space $\dR^2$
into two wedges $\Omega_1$ and $\Omega_{2}$ with angles of 
degrees $\varphi$ and $2\pi - \varphi$, see Figure~\ref{fig3}.
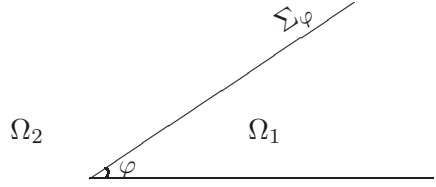
\begin{figure}[H]
\begin{center}
\begin{picture}(150,80)
\put(0,10){\line(1,0){130}}
\put(0,10){\line(3,2){100}}
\put(69,65){\begin{turn}{33.3}$\Sigma_\varphi$\end{turn}}
\qbezier(6,14)(8,13)(6,10)
\put(11,12){\small{$ \varphi$}}
\put(60,25){$\Omega_1$}
\put(-30,25){$\Omega_2$}
\end{picture}
\end{center}
\caption{The angle $\Sigma_\varphi$ separates the Euclidean space $\dR^2$ into two wedges $\Omega_1$ and $\Omega_2$.} 
\label{fig3}
\end{figure}
\noindent The underlying Hilbert space
can be decomposed as 
\[
L^2(\dR^2) = L^2(\Omega_1)\oplus L^2(\Omega_2).
\]
Any $f\in\dom\fra_{\varphi}$ can be written as the orthogonal sum $f_1\oplus f_2$
with respect to that decomposition of $L^2(\dR^2)$. 
Note that $f_1 \in H^1(\Omega_1)$ and that $f_{2} \in H^1(\Omega_{2})$.
Clearly,
\begin{equation}
\label{observation1}
\begin{split}
&\|f\|_{L^2(\dR^2)}^2 = \|f_1\|^2_{L^2(\Omega_1)} +
\|f_{2}\|^2_{L^2(\Omega_{2})},\\[0.2ex]
&\|\nabla f\|_{L^2(\dR^2;\dC^2)}^2 = \|\nabla f_1\|^2_{L^2(\Omega_1;\dC^2)} +
\|\nabla f_{2}\|^2_{L^2(\Omega_{2};\dC^2)}.
\end{split}
\end{equation}
The coupling constant can be decomposed as $\alpha = \beta + (\alpha - \beta)$ 
with some optimization parameter $\beta\in[0,\alpha]$ and the relation 
\begin{equation}
\label{observation2}
\alpha\|f|_{\Sigma_\varphi}\|^2_{L^2(\Sigma_\varphi)} = \beta\|f_1|_{\partial\Omega_1}\|^2_{L^2(\partial\Omega_1)} +(\alpha - \beta)\|f_{2}|_{\partial\Omega_2}\|^2_{L^2(\partial\Omega_2)}.
\end{equation}
holds. According to Proposition~\ref{wedge1}
\begin{equation}
\label{estimate1}
\|\nabla f_1\|^2_{L^2(\Omega_1;\dC^2)} - \beta\|f_1|_{\partial\Omega_1}\|^2_{L^2(\partial\Omega_1)} 
\ge -\tfrac{\beta^2}{\sin^2(\varphi/2)}\|f_1\|^2_{L^2(\Omega_1)},
\end{equation}
and according to Proposition~\ref{wedge2}
\begin{equation}
\label{estimate2}
\|\nabla f_{2}\|^2_{L^2(\Omega_{2};\dC^2)} - (\alpha - \beta)\|f_{2}|_{\partial\Omega_2}\|^2_{L^2(\partial\Omega_2)} 
\ge -(\alpha - \beta)^2\|f_{2}\|^2_{L^2(\Omega_{2})}.
\end{equation}
The observations \eqref{observation1}, \eqref{observation2} and the estimates \eqref{estimate1}, \eqref{estimate2} imply
\[
\fra_{\varphi}[f] \ge -\max\Big\{\tfrac{\beta^2}{\sin^2(\varphi/2)},(\alpha - \beta)^2\Big\}\|f\|^2_{L^2(\dR^2)}.
\]
Making optimization with respect to $\beta$, we observe that the maximum between the two values  in the estimate above is minimal, when these two values coincide. That is
\[
\tfrac{\beta^2}{\sin^2(\varphi/2)} = (\alpha - \beta)^2,
\]
which is equivalent to
\begin{equation}
\label{beta}
\beta = \tfrac{\alpha\sin(\varphi/2)}{(1+\sin(\varphi/2))},
\end{equation}
resulting in the final estimate
\[
\fra_{\varphi}[f] \ge -\tfrac{\alpha^2}{(1+\sin(\varphi/2))^2}\|f\|^2_{L^2(\dR^2)}.
\]
This final estimate implies the desired spectral bound.
\end{proof}
\begin{remark}
Note that the previously known lower bound \eqref{LLPbound} comes out from the proof of the last theorem if we choose $\beta = \alpha /2$, which is the optimal choice in our proof only for $\varphi = \pi$ as we see from \eqref{beta}.
\end{remark}

\section{A lower bound on the spectrum of $B_\varphi$}
In the next theorem we obtain a lower bound on the spectrum of the self-adjoint
operator $B_\varphi$ using the same idea as in Theorem~\ref{thm:lowerbound}.
\begin{thm}
\label{thm:lowerbound2}
Let the self-adjoint operator $B_\varphi$ be associated with the quadratic
form given in \eqref{qform2}. Then the estimate
\[
\inf\sigma(B_\varphi) \ge -\frac{\alpha^2}{1+\sin\varphi}
\]
holds.
\end{thm}
\begin{proof}
The crossing straight lines $\Gamma_1$ and $\Gamma_2$ separate the Euclidean space $\dR^2$
into four wedges $\{\Omega_k\}_{k=1}^4$. Namely, the wedges $\Omega_1$ and $\Omega_2$ with angles of  degree $\varphi$ and the wedges $\Omega_3$ and $\Omega_4$ with  angles of degree  $\pi - \varphi$, see Figure~\ref{fig4}.
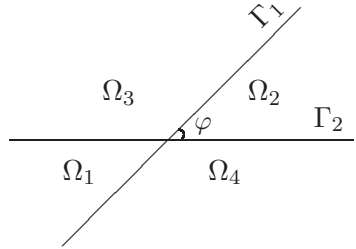
\begin{figure}[H]
\begin{center}
\begin{picture}(150,90)
\put(0,50){\line(1,0){130}}
\put(20,10){\line(1,1){90}}
\put(90,92){\begin{turn}{45.0}$\Gamma_1$\end{turn}}
\put(115,55){$\Gamma_2$}
\put(90,65){$\Omega_2$}
\put(35,65){$\Omega_3$}
\put(20,35){$\Omega_1$}
\put(75,35){$\Omega_4$}
\qbezier(64,54)(67,53)(65,50)
\put(70,54){\small{$ \varphi$}}
\end{picture}
\end{center}
\caption{The crossing straight lines $\Gamma_1$ and $\Gamma_2$ separate the Euclidean space $\dR^2$
into four wedges $\{\Omega_k\}_{k=1}^4$.}
\label{fig4}
\end{figure}
\noindent The underlying Hilbert space can be decomposed as 
\[
L^2(\dR^2) = \bigoplus_{k=1}^4 L^2(\Omega_k).
\]
Any $f\in\dom\frb_{\varphi}$ can be written as the orthogonal sum $\oplus_{k=1}^4 f_k$
with respect to that decomposition of $L^2(\dR^2)$. 
Note that $f_k \in H^1(\Omega_k)$ for $k=1,2,3,4$. 
Clearly,
\begin{equation}
\label{observation3}
\|f\|_{L^2(\dR^2)}^2 = \sum_{k=1}^4\|f_k\|^2_{L^2(\Omega_k)},\quad \|\nabla f\|_{L^2(\dR^2;\dC^2)}^2 = \sum_{k=1}^4\|\nabla f_k\|^2_{L^2(\Omega_k;\dC^2)}.
\end{equation}
The coupling constant can be decomposed as $\alpha = \beta + (\alpha - \beta)$ 
with some optimization parameter $\beta\in[0,\alpha]$ and the relation
\begin{equation}
\label{observation4}
\begin{split}
\alpha\|f|_{\Gamma_\varphi}\|^2_{L^2(\Gamma_\varphi)} &= \beta\|f_1|_{\partial\Omega_1}\|^2_{L^2(\partial\Omega_1)} +  \beta\|f_2|_{\partial\Omega_2}\|^2_{L^2(\partial\Omega_2)} \\ &\quad+(\alpha - \beta)\|f_3|_{\partial\Omega_3}\|^2_{L^2(\partial\Omega_3)}
+(\alpha - \beta)\|f_4|_{\partial\Omega_4}\|^2_{L^2(\partial\Omega_4)}
\end{split}
\end{equation}
holds. According to Proposition~\ref{wedge1}
\begin{equation}
\label{estimate3}
\begin{split}
\|\nabla f_1\|^2_{L^2(\Omega_1;\dC^2)} - \beta\|f_1|_{\partial\Omega_1}\|^2_{L^2(\partial\Omega_1)} 
&\ge -\tfrac{\beta^2}{\sin^2(\varphi/2)}\|f_1\|^2_{L^2(\Omega_1)},\\[0.2ex]
\|\nabla f_2\|^2_{L^2(\Omega_2;\dC^2)} - \beta\|f_2|_{\partial\Omega_2}\|^2_{L^2(\partial\Omega_2)} 
&\ge -\tfrac{\beta^2}{\sin^2(\varphi/2)}\|f_2\|^2_{L^2(\Omega_2)}.\\
\end{split}
\end{equation}
Also according to Proposition~\ref{wedge1}
\begin{equation}
\label{estimate4}
\begin{split}
\|\nabla f_3\|^2_{L^2(\Omega_3;\dC^2)} - (\alpha - \beta)\|f_3|_{\partial\Omega_3}\|^2_{L^2(\partial\Omega_3)} 
&\ge -\tfrac{(\alpha -\beta)^2}{\cos^2(\varphi/2)}\|f_3\|^2_{L^2(\Omega_3)},\\
\|\nabla f_4\|^2_{L^2(\Omega_4;\dC^2)} - (\alpha - \beta)\|f_4|_{\partial\Omega_4}\|^2_{L^2(\partial\Omega_4)} 
&\ge -\tfrac{(\alpha -\beta)^2}{\cos^2(\varphi/2)}\|f_4\|^2_{L^2(\Omega_4)}.\\
\end{split}
\end{equation}
The observations \eqref{observation3}, \eqref{observation4} and the estimates \eqref{estimate3}, \eqref{estimate4}  imply
\[
\frb_{\varphi}[f] \ge -\max\Big\{\tfrac{\beta^2}{\sin^2(\varphi/2)},\tfrac{(\alpha - \beta)^2}{\cos^2(\varphi/2)}\Big\}\|f\|^2_{L^2(\dR^2)}.
\]
Making optimization with respect to $\beta$, we observe that the maximum between the two values in the estimate above is minimal, when these two values coincide. That is
\[
\tfrac{\beta^2}{\sin^2(\varphi/2)} = \tfrac{(\alpha - \beta)^2}{\cos^2(\varphi/2)},
\]
which is equivalent to
\begin{equation}
\label{beta2}
\beta = \tfrac{\alpha\tan(\varphi/2)}{(1+\tan(\varphi/2))},
\end{equation}
resulting in the final estimate
\[
\frb_{\varphi}[f] \ge -\tfrac{\alpha^2}{1+\sin(\varphi)}\|f\|^2_{L^2(\dR^2)}.
\]
This final estimate implies the desired spectral bound.
\end{proof}
\begin{remark}
The result of Theorem~\ref{thm:lowerbound2} complements \cite[Theorem 4.6\,(iv)]{CDR08}, where the bound
\[
\inf\sigma(B_\varphi) \ge -\alpha^2.
\]
for all $\varphi\in(0,\pi)$ was obtained.
\end{remark}
\section{Acknowledgements}
The author is grateful to Jussi Behrndt, Sylwia Kondej, Igor Lobanov, Igor Yu. Popov, and Jonathan Rohleder for discussions.
The work was supported by Austrian Science Fund (FWF): project P 25162-N26 and 
partially supported by the Ministry of Education and Science of
Russian Federation: project 14.B37.21.0457.

\end{document}